\documentclass[copyright, creativecommons]{eptcsstyle/eptcs}
\providecommand{\event}{QPL 2018} % Name of the event you are submitting to
\usepackage{xcolor}
\hypersetup{
    colorlinks,
    linkcolor={red!50!black},
    citecolor={blue!50!black},
    urlcolor={blue!80!black}
}

\usepackage[utf8]{inputenc}
\newcommand{\upsidedownbang}{\rotatebox[origin=c]{180}{!}}
\usepackage{anyfontsize}
\usepackage{xspace}
\usepackage{float}

\usepackage{wrapfig}
\usepackage{subcaption}

\usepackage[inline,shortlabels]{enumitem}

\usepackage{mathpartir}

\usepackage{cmll} % Linear logic symbols
\usepackage{ stmaryrd } % more math symbols
\usepackage{upquote} % straight quote

\usepackage[numbers]{natbib}
\usepackage[qm,braket]{qcircuit}
\newcommand{\qwire}{\ensuremath{\mathcal{Q}\textsc{wire}}\xspace}

\newcommand{\revs}{R\textsc{evs}\xspace}
\newcommand{\reverc}{R\textsc{e}V\textsc{er}C\xspace}

\newcommand{\qnot}[2]{\draw (#1,#2) circle [radius=0.25]; \draw (#1,#2-0.25) -- (#1,#2+0.25);}
\newcommand{\cnot}[3]{\draw[fill=black] (#1,#2) circle [radius=0.13]; \draw (#1,#3) circle [radius=0.25]; \draw (#1,#2) -- (#1,#3+0.25);}
\newcommand{\tof}[4]{\draw[fill=black] (#1,#2) circle [radius=0.13]; \draw[fill=black] (#1,#3) circle [radius=0.13]; \draw (#1,#4) circle [radius=0.25]; \draw (#1,#2) -- (#1,#4+0.25);}

\newcommand{\TRUE}{t}
\newcommand{\FALSE}{f}
\usepackage{tikz} %commutativity diagrams
\usetikzlibrary{%
  arrows,%
  shapes.misc,% wg. rounded rectangle
  shapes.arrows,%
  shapes.callouts,
  chains,%
  matrix,%
  positioning,% wg. " of "
  scopes,%
  decorations.pathmorphing,% /pgf/decoration/random steps | erste Graphik
  decorations.text,
  shadows%
}
\tikzset{ machine/.style={
    rectangle,
    minimum width=20mm,
    minimum height=10mm,
    text width=16mm,
    align=center,
    very thick,
    draw=black,
    color=black,
    fill=white,
    font=\ttfamily,
  }
}

\newcommand{\etal}{\emph{et al.}\xspace}

\newcommand{\ie}{\emph{i.e.,}\xspace}

\newcommand{\lbind}{\ensuremath{\mathbin{=\!\!>\!\!>\!\!=}}}

\newcommand{\denote}[1]{\ensuremath{\llbracket#1\rrbracket}}

\usepackage{booktabs}

\newcommand{\One}{\code{One}}
\newcommand{\Bit}{\code{Bit}}
\newcommand{\Qubit}{\code{Qubit}}
\usepackage{listings,lstcoq}
\usepackage{MnSymbol}
\definecolor{ltblue}{rgb}{0,0.4,0.4}
\definecolor{dkblue}{rgb}{0,0.1,0.6}
\definecolor{dkgreen}{rgb}{0,0.35,0}
\definecolor{dkviolet}{rgb}{0.3,0,0.5}
\definecolor{dkred}{rgb}{0.5,0,0}
\newcommand{\code}[1]{\textup{\texttt{#1}}}
\usepackage{newunicodechar}
\let\Alpha=A
\let\Beta=B
\let\Epsilon=E
\let\Zeta=Z
\let\Eta=H
\let\Iota=I
\let\Kappa=K
\let\Mu=M
\let\Nu=N
\let\Omicron=O
\let\omicron=o
\let\Rho=P
\let\Tau=T
\let\Chi=X

\newunicodechar{Α}{\ensuremath{\Alpha}}
\newunicodechar{α}{\ensuremath{\alpha}}
\newunicodechar{Β}{\ensuremath{\Beta}}
\newunicodechar{β}{\ensuremath{\beta}}
\newunicodechar{Γ}{\ensuremath{\Gamma}}
\newunicodechar{γ}{\ensuremath{\gamma}}
\newunicodechar{Δ}{\ensuremath{\Delta}}
\newunicodechar{δ}{\ensuremath{\delta}}
\newunicodechar{Ε}{\ensuremath{\Epsilon}}
\newunicodechar{ε}{\ensuremath{\epsilon}}
\newunicodechar{ϵ}{\ensuremath{\varepsilon}}
\newunicodechar{Ζ}{\ensuremath{\Zeta}}
\newunicodechar{ζ}{\ensuremath{\zeta}}
\newunicodechar{Η}{\ensuremath{\Eta}}
\newunicodechar{η}{\ensuremath{\eta}}
\newunicodechar{Θ}{\ensuremath{\Theta}}
\newunicodechar{θ}{\ensuremath{\theta}}
\newunicodechar{ϑ}{\ensuremath{\vartheta}}
\newunicodechar{Ι}{\ensuremath{\Iota}}
\newunicodechar{ι}{\ensuremath{\iota}}
\newunicodechar{Κ}{\ensuremath{\Kappa}}
\newunicodechar{κ}{\ensuremath{\kappa}}
\newunicodechar{Λ}{\ensuremath{\Lambda}}
\newunicodechar{λ}{\ensuremath{\lambda}}
\newunicodechar{Μ}{\ensuremath{\Mu}}
\newunicodechar{μ}{\ensuremath{\mu}}
\newunicodechar{Ν}{\ensuremath{\Nu}}
\newunicodechar{ν}{\ensuremath{\nu}}
\newunicodechar{Ξ}{\ensuremath{\Xi}}
\newunicodechar{ξ}{\ensuremath{\xi}}
\newunicodechar{Ο}{\ensuremath{\Omicron}}
\newunicodechar{ο}{\ensuremath{\omicron}}
\newunicodechar{Π}{\ensuremath{\Pi}}
\newunicodechar{π}{\ensuremath{\pi}}
\newunicodechar{ϖ}{\ensuremath{\varpi}}
\newunicodechar{Ρ}{\ensuremath{\Rho}}
\newunicodechar{ρ}{\ensuremath{\rho}}
\newunicodechar{ϱ}{\ensuremath{\varrho}}
\newunicodechar{Σ}{\ensuremath{\Sigma}}
\newunicodechar{σ}{\ensuremath{\sigma}}
\newunicodechar{ς}{\ensuremath{\varsigma}}
\newunicodechar{Τ}{\ensuremath{\Tau}}
\newunicodechar{τ}{\ensuremath{\tau}}
\newunicodechar{Υ}{\ensuremath{\Upsilon}}
\newunicodechar{υ}{\ensuremath{\upsilon}}
\newunicodechar{Φ}{\ensuremath{\Phi}}
\newunicodechar{φ}{\ensuremath{\phi}}
\newunicodechar{ϕ}{\ensuremath{\varphi}}
\newunicodechar{Χ}{\ensuremath{\Chi}}
\newunicodechar{χ}{\ensuremath{\chi}}
\newunicodechar{Ψ}{\ensuremath{\Psi}}
\newunicodechar{ψ}{\ensuremath{\psi}}
\newunicodechar{Ω}{\ensuremath{\Omega}}
\newunicodechar{ω}{\ensuremath{\omega}}

\newunicodechar{ℕ}{\ensuremath{\mathbb{N}}}
\newunicodechar{∅}{\ensuremath{\emptyset}}

\newunicodechar{•}{\ensuremath{\bullet}}
\newunicodechar{≈}{\ensuremath{\approx}}
\newunicodechar{≅}{\ensuremath{\cong}}
\newunicodechar{≡}{\ensuremath{\equiv}}
\newunicodechar{≤}{\ensuremath{\le}}
\newunicodechar{≥}{\ensuremath{\ge}}
\newunicodechar{≠}{\ensuremath{\neq}}
\newunicodechar{∀}{\ensuremath{\forall}}
\newunicodechar{∃}{\ensuremath{\exists}}
\newunicodechar{±}{\ensuremath{\pm}}
\newunicodechar{∓}{\ensuremath{\pm}}
\newunicodechar{·}{\ensuremath{\cdot}}
\newunicodechar{⋯}{\ensuremath{\cdots}}
\newunicodechar{…}{\ensuremath{\ldots}}
\newunicodechar{∷}{~\mathrel{:\!\!\!:}~}
\newunicodechar{×}{\ensuremath{\times}}
\newunicodechar{∞}{\ensuremath{\infty}}
\newunicodechar{→}{\ensuremath{\to}}
\newunicodechar{←}{\ensuremath{\leftarrow}}
\newunicodechar{⇒}{\ensuremath{\Rightarrow}}
\newunicodechar{↦}{\ensuremath{\mapsto}}
\newunicodechar{↝}{\ensuremath{\leadsto}}
\newunicodechar{∨}{\ensuremath{\vee}}
\newunicodechar{∧}{\ensuremath{\wedge}}
\newunicodechar{⊢}{\ensuremath{\vdash}}
\newunicodechar{⊣}{\ensuremath{\dashv}}
\newunicodechar{∣}{\ensuremath{\mid}}
\newunicodechar{∈}{\ensuremath{\in}}
\newunicodechar{⊂}{\ensuremath{\subset}}
\newunicodechar{⊆}{\ensuremath{\subseteq}}
\newunicodechar{⋓}{\ensuremath{\Cup}}
\newunicodechar{∉}{\ensuremath{\not\in}}

\newunicodechar{⊸}{\ensuremath{\multimap}}
\newunicodechar{¬}{\ensuremath{\neg}}
\newunicodechar{⊗}{\ensuremath{\otimes}}
\newunicodechar{⨂}{\ensuremath{\bigotimes}}
\newunicodechar{⊕}{\ensuremath{\oplus}}
\newunicodechar{〈}{\ensuremath{\langle}}
\newunicodechar{⟨}{\ensuremath{\langle}}
\newunicodechar{⟩}{\ensuremath{\rangle}}
\newunicodechar{〉}{\ensuremath{\rangle}}
\newunicodechar{⌈}{\ensuremath{\lceil}}
\newunicodechar{⌉}{\ensuremath{\rceil}}
\newunicodechar{¡}{\ensuremath{\upsidedownbang}}
\newunicodechar{∘}{\ensuremath{\circ}}
\newunicodechar{†}{\ensuremath{\dagger}}
\newunicodechar{⊤}{\ensuremath{\top}}
\newunicodechar{∥}{\ensuremath{\|}}

\newunicodechar{〚}{\ensuremath{\llbracket}}
\newunicodechar{〛}{\ensuremath{\rrbracket}}
\newunicodechar{⟦}{\ensuremath{\llbracket}}
\newunicodechar{⟧}{\ensuremath{\rrbracket}}

\usepackage{etoolbox} % replaces ifthen package
\newtoggle{comments}
\toggletrue{comments}
\iftoggle{comments}{
  \newcommand{\jennifer}[1]{\textbf{\textcolor{red}{[ #1 --- Jennifer]}}}
  \newcommand{\steve}[1]{\textbf{\textcolor{purple}{[ #1 --- Steve]}}}
  \newcommand{\robert}[1]{\textbf{\textcolor{blue}{[ #1 --- Robert]}}}

  \newcommand{\fixme}[1]{\textbf{\textcolor{red}{[ Fixme: #1]}}}
  \newcommand{\note}[1]{\textbf{\textcolor{green}{[ Note: #1 ]}}}
  \newcommand{\todo}[1]{\textbf{\textcolor{green}{[ TODO: #1 ]}}}
}{
  \newcommand{\jennifer}[1]{}
  \newcommand{\steve}[1]{}
  \newcommand{\robert}[1]{}

  \newcommand{\fixme}[1]{}
  \newcommand{\note}[1]{}
  \newcommand{\todo}[1]{}
}
\usepackage{amsthm}
\usepackage[capitalize,noabbrev]{cleveref} % must be loaded after hyperref
\newtheorem{theorem}{Theorem}[section]

\newtheorem{lemma}[theorem]{Lemma}
\newtheorem{conjecture}[theorem]{Conjecture}
\usepackage[bottom]{footmisc}

\title{\textit{Re}\qwire: Reasoning about Reversible Quantum Circuits
\thanks{This work is supported in part by AFOSR MURI No. FA9550-16-1-0082}}
\def\titlerunning{\textit{Re}\qwire}

\author{Robert Rand \qquad Jennifer Paykin \qquad  \qquad Dong-Ho Lee \qquad Steve Zdancewic \\
\{rrand, jpaykin, dongle, stevez\}@cis.upenn.edu \institute{~\\ \Large University of Pennsylvania}}

\def\authorrunning{Rand, Paykin, Lee \& Zdancewic}
\renewcommand\footnotemark{} % Remove thanks asterisk

\begin{document}

\maketitle

%%%%%%%%%%%%%%%%%%
%% PAPER STARTS %%
%%%%%%%%%%%%%%%%%%

\begin{abstract}
Common quantum algorithms make heavy use of ancillae: scratch qubits that are initialized at some state and later returned to that state and discarded. Existing quantum circuit languages let programmers assert that a qubit has been returned to the $\ket{0}$ state before it is discarded, allowing for a range of optimizations. However, existing languages do not provide the tools to verify these assertions, introducing a potential source of errors. In this paper we present methods for verifying that ancillae are discarded in the desired state, and use these methods to implement a verified compiler from classical functions to quantum oracles.

\end{abstract}

\section{Introduction}
\label{sec:introduction}

\begin{figure}
\begin{center} \begin{tikzpicture}[scale=0.8, x = 0.8cm, y=-0.8cm]

    % Wires
    \draw (0,0) node[left] {$a$} -- (6,0) node[right] {$a$};
    \draw (0,1) node[left] {$b$} -- (6,1) node[right] {$b$};
    \draw (0,2) node[left] {$z$} -- (6,2) node[right] {$z \oplus (a \wedge b)$};

    \draw (12,0) node[left] {$a$} -- (18,0) node[right] {$a$};
    \draw (12,1) node[left] {$b$} -- (18,1) node[right] {$b$};
    \draw (12,2) node[left] {$z$} -- (18,2) node[right] {$z \oplus (a \vee b)$};

    % Gates     
    \tof{3}{0}{1}{2}

    \qnot{13.5}{0}
    \qnot{13.5}{1}
    \tof{15}{0}{1}{2} 
    \qnot{16.5}{0}
    \qnot{16.5}{1}
    \qnot{16.5}{2}

\end{tikzpicture} \end{center}
\caption{Quantum oracles implementing the boolean $\wedge$ and $\vee$. The $\oplus$ gates represent negation, and $\bullet$ represents control.}
\label{fig:andor}
\end{figure}
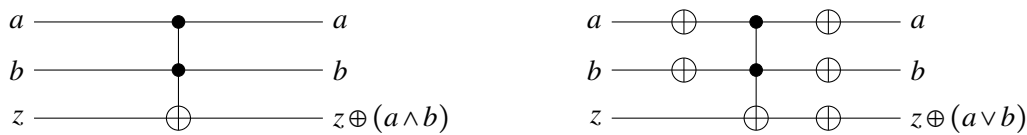

Many quantum algorithms rely heavily on \emph{quantum oracles}, classical programs 
executed inside quantum circuits. Toffoli proved that any
classical, boolean-valued function $f(x)$ can be implemented as a unitary
circuit $f_u$ satisfying $f_u(x,z) = (x,z ⊕ f(x))$~\cite{Toffoli1980}. Toffoli's
construction for quantum oracles is used in many quantum algorithms, such as the
modular arithmetic of Shor's algorithm~\cite{Shor1999}. As a concrete example,
\cref{fig:andor} shows quantum circuits that implement the boolean functions
and ($\wedge$) and or ($\vee$).

Unfortunately, Toffoli's construction introduces significant overhead. Consider
a circuit meant to compute the boolean formula $(a \vee b) \wedge (c \vee d)$.
The circuit needs two additional scratch wires, or \emph{ancillae}, to carry the
outputs of $(a \vee b)$ and $(c \vee d)$, as seen in \cref{fig:sat}. The
annotation $0$ at the start of a wire means that qubit is initialized in the
state $\ket{0}$. When constructed in this naive way, the resulting circuit no
longer corresponds to a unitary
transformation, %due to the initialization gates,
and cannot be safely used in a larger quantum circuit. 
%\jennifer{Is this still
%  Toffoli's construction? Is Toffoli's construction for quantum or classical
%  circuits? Does it include uncomputing?}

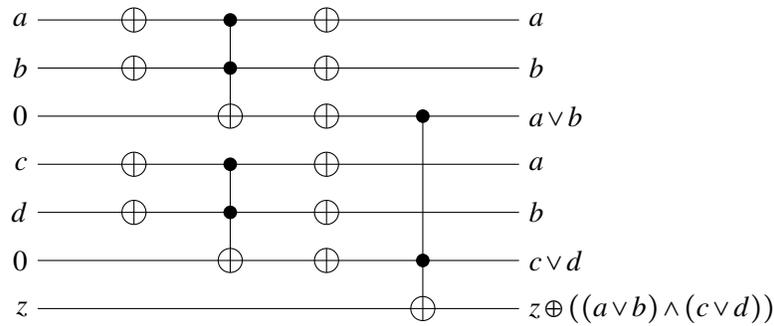
\begin{figure}
\begin{center} \begin{tikzpicture}[scale=0.8, x = 0.8cm, y=-0.8cm]

    % Wires
    \draw (0,0) node[left] {$a$} -- (10,0) node[right] {$a$};
    \draw (0,1) node[left] {$b$} -- (10,1) node[right] {$b$};
    \draw (0,2) node[left] {$0$} -- (10,2) node[right] {$a \vee b$};
    \draw (0,3) node[left] {$c$} -- (10,3) node[right] {$a$};
    \draw (0,4) node[left] {$d$} -- (10,4) node[right] {$b$};
    \draw (0,5) node[left] {$0$} -- (10,5) node[right] {$c \vee d$};
    \draw (0,6) node[left] {$z$} -- (10,6) node[right] {$z \oplus ((a \vee b) \wedge (c \vee d))$};

    % Gates     
    \qnot{2}{0}
    \qnot{2}{1}
    \qnot{2}{3}
    \qnot{2}{4}    
    \tof{4}{0}{1}{2}
    \tof{4}{3}{4}{5}
    \qnot{6}{0}
    \qnot{6}{1}
    \qnot{6}{2}
    \qnot{6}{3}
    \qnot{6}{4}    
    \qnot{6}{5}
    \tof{8}{2}{5}{6}

\end{tikzpicture} \end{center}
\caption{An non-unitary quantum oracle for $(a \vee b) \wedge (c \vee d)$}
\label{fig:sat}
\end{figure}

The solution is to \emph{uncompute} the intermediate values $a \vee b$ and
$c \vee d$ and then discard them at the end of the quantum circuit
(\cref{fig:rsat}). The annotation $0$ at the end of a wire is an
\emph{assertion} that the qubit at that point is in the zero state, at which
point we can safely discard it without affecting the remainder of the state.
(If we measured and discarded a non-zero qubit, we would affect whatever qubits it was entangled with.)

How can we verify that such an assertion is actually true? We cannot dynamically
check the assertion, since we can only access the value of a qubit by measuring
it, collapsing the qubit in question to a $0$ or $1$ state. However, we can
statically reason that the qubit must be in the state $\ket{0}$ by analyzing the
circuit semantics.

\begin{figure}
\begin{center} \begin{tikzpicture}[scale=0.8, x = 0.8cm, y=-0.8cm]

    % Wires
    \draw (0,0) node[left] {$a$} -- (16,0) node[right] {$a$};
    \draw (0,1) node[left] {$b$} -- (16,1) node[right] {$b$};
    \draw (1,2) node[left] {$0$} -- (15,2) node[right] {$0$};
    \draw (0,3) node[left] {$c$} -- (16,3) node[right] {$a$};
    \draw (0,4) node[left] {$d$} -- (16,4) node[right] {$b$};
    \draw (1,5) node[left] {$0$} -- (15,5) node[right] {$0$};
    \draw (0,6) node[left] {$z$} -- (14,6) node[right] {$z \oplus ((a \vee b) \wedge (c \vee d))$};

    % Gates     
    \qnot{2}{0}
    \qnot{2}{1}
    \qnot{2}{3}
    \qnot{2}{4}    
    \tof{4}{0}{1}{2}
    \tof{4}{3}{4}{5}
    \qnot{6}{0}
    \qnot{6}{1}
    \qnot{6}{2}
    \qnot{6}{3}
    \qnot{6}{4}    
    \qnot{6}{5}
    \tof{8}{2}{5}{6}
    \qnot{10}{0}
    \qnot{10}{1}
    \qnot{10}{2}
    \qnot{10}{3}
    \qnot{10}{4}    
    \qnot{10}{5}
    \tof{12}{0}{1}{2}
    \tof{12}{3}{4}{5}
    \qnot{14}{0}
    \qnot{14}{1}
    \qnot{14}{3}
    \qnot{14}{4}    

\end{tikzpicture} \end{center}
\caption{A unitary quantum oracle for $(a \vee b) \wedge (c \vee d)$ with ancillae}
\label{fig:rsat}
\end{figure}
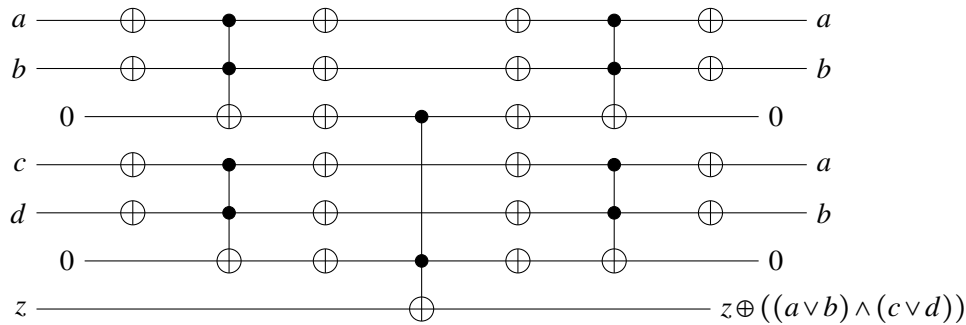

The claim that a qubit is in the $0$ state is a \emph{semantic} assertion about the
behavior of the circuit. Unfortunately, this makes it hard to verify---computing
the semantics of a quantum program is computationally intractable in the general case.
Circuit programming languages often allow users to make such assertions but not
to verify that they are true. For example, Quipper~\cite{Green2013} allows
programmers to make assertions about the state of ancillae, but these assertions
are never checked. Likewise, in Q\# \cite{Svore2018} the assertion will be checked by
a simulator but cannot be checked when a program is run on a quantum computer.
Hence, when the qubit is reused, a common use for ancillae which Q\# emphasizes, it may be in the wrong state. The QCL quantum
circuit language~\cite{Omer2005} provides a
built-in method for creating reversible circuits from classical functions, but
the programmer must trust this method to safely manage ancillae. In a step in the
right direction, the $\reverc$
compiler~\cite{Amy2017} for the (non-quantum) reversible computing language
$\revs$~\cite{Parent2017} provides a similar approach to compilation and verifies that it
correctly uncomputes its ancilla. However, other assertions in $\revs$ that a
wire is correctly in the $0$ state are ignored if they cannot be automatically
verified.

In this paper, we develop verification techniques for safely working with
ancillae. Our approach allows the programmer to discard qubits that are in the
state $\ket{0}$ or $\ket{1}$, provided that she first formally proves that the
qubits are in the specified state. Inspired by the $\reverc$
compiler~\cite{Amy2017}, we also provide syntactic conditions that the
programmer may satisfy to guarantee that her assertions are true. However,
circuits do not need to match this syntactic specification: a programmer may
instead manually prove that her circuit safely discards qubits using the
denotational semantics of the language. This gives the programmer the
flexibility to use ancillae even when the proofs of such assertions are
non-trivial.

We develop these techniques in the context of \qwire (as in ``require''), a
domain-specific programming language for describing and reasoning about quantum
circuits~\cite{Paykin2017}. \qwire is implemented as a embedded language inside
the Coq proof assistant~\cite{Coq}, which allows us to formally verify
properties of \qwire circuits. These properties can range from coarse-grained
(``this circuit corresponds to a unitary transformation'') to precise (``this
teleport circuit is equal to the identity'')~\cite{Rand2017}. \qwire is an
ongoing project and available for public use at
\url{https://github.com/inQWIRE/QWIRE}.

This paper reports on work-in-progress that makes the following
contributions:\footnote{\cref{sec:related} elaborates on the state of the Coq
  development that underlies this work.}
\begin{itemize}
\item We extend \qwire with assertion-bearing ancillae.
\item We give semantic conditions for the closely related properties of (a) when
  a circuit is reversible and (b) when a
  circuit contains only valid assertions about its ancillae.
\item We provide syntactic conditions that guarantee the correctness of assertions with respect to these semantic conditions.
\item We present a compiler to turn boolean expressions into reversible \qwire circuits and we prove its correctness using our syntactic correctness conditions.
\item Finally, we use our verified compiler to implement a quantum adder.
\end{itemize}

%\section{outline}
%\begin{enumerate}
%\item Introduction 
%\begin{enumerate}
%\item Introduction to \qwire
%\item New features of \qwire (solve\_matrix, composition, deutsch and teleport)
%\end{enumerate}
%\item Ancilla management
%\begin{enumerate}
%\item \emph{assert0} and \emph{assert1} gates
%\item The \emph{unsafe} denotation function
%\item Validity
%\begin{enumerate}
%\item Semantic validity (safe = unsafe and trace = 1)
%\item Syntactic validity guarantees
%\end{enumerate}
%\item Pruning ancillae
%\end{enumerate}
%\item Reversible circuits
%\begin{enumerate}
%\item Use of oracles
%\item Running example (? TODO)
%\item Compiling boolean expressions to circuits 
%\item Compiling complex programs to circuits (? TODO)  
%\item Verifying correctness
%\end{enumerate}
%\item Related Work
%\begin{enumerate}
%\item \reverc
%\item Quipper
%\end{enumerate}
%\end{enumerate}

%%% Local Variables:
%%% mode: latex
%%% TeX-master: "main"
%%% End:

\section{The \qwire Circuit Language}
\label{sec:qwire}
\qwire~\cite{Paykin2017} is a small quantum circuit language designed to be
embedded in a larger, functional programming language. We have implemented
\qwire in the Coq proof assistant, which provides access to dependent types and
the Coq interactive proof system. We use these features to type check \qwire
circuits and verify properties about their semantics~\cite{Rand2017}. In this
section we give a brief introduction to the syntax and semantics of \qwire,
including the new ancilla assertions.

A \qwire circuit consists of a sequence of gate applications terminated with some output wires\footnote{For simplicity, this presentation elides a communication protocol called \emph{dynamic lifting} discussed in prior work on \qwire~\cite{Paykin2017,Rand2017}.}.
%\[ \mathtt{Circuit} \; W \ ::= \ \mathtt{output} \; p  \mid \mathtt{gate} \; p' \leftarrow \; g \; p; \; \mathtt{Circuit} \; W \] 
\begin{center}
\begin{tabular}{c}
\begin{coq}
Circuit W ::= output p  | gate p' <- g p ; Circuit W
\end{coq}
\end{tabular}
\end{center}

The parameter \coqe{W} refers to a \emph{wire type}: \coqe{Bit}, \coqe{Qubit} or
some tuple of \coqe{Bit}s and \coqe{Qubit}s (including the empty tuple
\coqe{One}). A \emph{pattern of wires}, denoted \coqe{p}, can be a bit-valued %
wire~\coqe{bit v}, a qubit-valued wire \coqe{qubit v}, a pair of wires
\coqe{(p1,p2)} or an empty tuple \coqe{()}.
% Wire variables \coqe{v} are represented by natural numbers. 
Gates \coqe{g} are either unitary gates \coqe{U}, drawn from a universal gate
set,
%(the Clifford + T set, along with real-valued rotation gates) 
or members of a small set of non-unitary gates:
% \[ g  := U \mid init_0 \mid init_1 \mid meas \mid discard \mid assert_0 \mid assert_1 \]

\begin{center}
\begin{tabular}{c}
 \begin{coq}
 W := Bit | Qubit | One | W ⊗ W
 g := U | init_0 | init_1 | meas | discard | assert_0 | assert_1
 \end{coq}
\end{tabular}
\end{center}
The \coqe{init} and \coqe{meas} gates initialize and measure qubits,
respectively; \coqe{meas} results in a bit, which can be discarded by the
\coqe{discard} bit or used as a control. The \coqe{assert_0} and \coqe{assert_1}
gates take a qubit as input and discard it, provided that it is in the state
$\ket{0}$ or $\ket{1}$ respectively. We will discuss the semantics of these
gates, and how to verify assertions, in \cref{sec:semantics,sec:guarantees}.

As an example, the following \qwire circuit prepares a Bell state:
\begin{center}
 \begin{minipage}{0.35\textwidth}
 \begin{coq}
gate p1 <- init_0 ();
gate p2 <- init_0 ();
gate p1 <- H p1;
gate (p1, p2) <- CNOT (p1, p2);
output (p1,p2)
\end{coq}
\end{minipage}
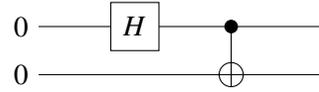
\begin{minipage}{0.6\textwidth}
\begin{center} \begin{tikzpicture}[scale=0.8, x = 0.8cm, y=-0.8cm]

    % Wires
    \draw (0,0) node[left] {$0$} -- (6,0);
    \draw (0,1) node[left] {$0$} -- (6,1);
            
    % Gates     
    \cnot{4}{0}{1}
    \draw[fill=white] (1.5,-0.5) rectangle node {$H$} (2.5,0.5);
    
\end{tikzpicture} \end{center}

\end{minipage}
\end{center}

\qwire also includes some more powerful functionality for constructing circuits.
Circuits can be composed via a let binding \coqe{let_ p <- C; C'}, where the
output of the first circuit \coqe{C} is plugged into the wires \coqe{p} in
\coqe{C'}. It's worth highlighting two useful instances of composition: The
\coqe{inSeq} (\coqe{;;}) function takes a \coqe{Box W W'} and a %
\coqe{Box W' W''} and composes them sequentially to return a \coqe{Box W W''}.
The inPar function likewise takes a circuit \coqe{c1} of type \coqe{Box W1 W1'}
and \coqe{c2} of type \coqe{Box W2 W2'} and composes them in parallel, producing
\coqe{c1 ∥ c2} of type \coqe{Box (W1 ⊗ W2) (W1' ⊗ W2')}.

%\emph{Dynamic lifting} \coqe{lift x <- p; C} is a \qwire primitive that lifts a
%bit-valued wire $p$ into a Coq boolean variable $x$, which can be used to
%compute a continuation of the circuit $C$. 
Circuits can also be \emph{boxed} by collecting the input of a circuit in an
input pattern \coqe{box_ p => C}, creating a closed term of type \coqe{Box W W'}
in the host language. Here, the input wire type \coqe{W} matches the type of the
input wire \coqe{p}, and the output type \coqe{W'} is the same as that of the
underlying circuit. Such a boxed circuit can be \emph{unboxed} to be used again
in other circuits.

Boxing, unboxing, and composing circuits is illustrated by the \emph{teleport}
circuit, where \coqe{alice}, \coqe{bob}, and \coqe{bell00} are simple circuits whose definitions are not shown:

% No dynamic lifting
 \begin{minipage}{0.38\textwidth}
\begin{coq}
Definition teleport : Box Qubit Qubit :=
  box_ q =>
    let_ (a,b) <- unbox bell00 () ;
    let_ (x,y) <- unbox alice (q,a) ;
    unbox bob (x,y,b).
 \end{coq}
\end{minipage}
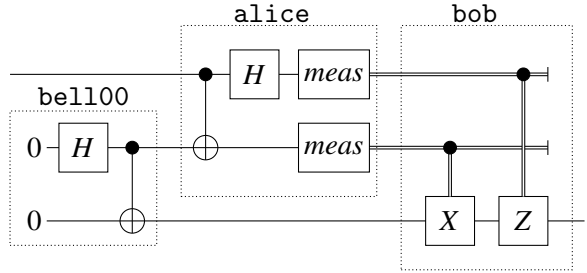
\begin{minipage}{0.65\textwidth}
\begin{center} \begin{tikzpicture}[scale=0.65]
    % Wires
    \draw (1.75,2) -- (12.75,2);
    \draw (1.75,3.5) -- (8,3.5);
          \draw[style=double] (8,3.5) -- (12,3.5);
          \draw (12,3.35) -- (12,3.65);
    \draw (1,5) -- (8,5);
          \draw[style=double] (8,5) -- (12,5);
          \draw (12,4.85) -- (12,5.15);

    % Bell00
    \draw[densely dotted] (1,1.5) rectangle  (4,4.25);
    \draw (2.5,4.55) node {\code{bell00}};

    \draw (1.5,2) node {$0$};
    \draw (1.5,3.5) node {$0$};

    \draw[fill=white] (2,3) rectangle node {$H$} (3,4);
    
    \draw[fill=black] (3.5,3.5) circle [radius=0.13];
    \draw (3.5,3.5) -- (3.5,1.75);
    \draw (3.5,2) circle [radius=0.25];

    % Alice
    \draw[densely dotted] (4.5,2.5) rectangle (8.5,6);
    \draw (6.35,6.25) node {\code{alice}};

    \draw[fill=black] (5,5) circle [radius=0.13];
    \draw (5,5) -- (5,3.25);
    \draw (5,3.5) circle [radius=0.25];

    \draw[fill=white] (5.5,4.5) rectangle node {$H$} (6.5,5.5);

    \draw[fill=white] (6.9,4.5) rectangle node {$meas$} (8.35,5.5);
    \draw[fill=white] (6.9,3) rectangle node {$meas$} (8.35,4);

    % Bob
    \draw[densely dotted] (9,1) rectangle (12.5,6);
    \draw (10.5,6.25) node {\code{bob}};

    \draw[fill=white] (9.5,1.5) rectangle node {$X$} (10.5,2.5);
    \draw[fill=white] (11,1.5) rectangle node {$Z$} (12,2.5);

    \draw[style=double] (10,3.5) -- (10,2.5);
    \draw[fill=black] (10,3.5) circle [radius=0.13];

    \draw[style=double] (11.5,5) -- (11.5,2.5);
    \draw[fill=black] (11.5,5) circle [radius=0.13];

\end{tikzpicture} \end{center}

\end{minipage}

\section{A \emph{Safe} and \emph{Unsafe} Semantics}
\label{sec:semantics}
As in prior work, \qwire's semantics is given in terms of \emph{density
  matrices}, denoted $\rho$, that represent distributions over pure quantum
states known as \emph{mixed states}. A \qwire circuit of type \coqe{Box W W'}
maps a $2^{⟦W⟧} × 2^{⟦W⟧}$ density matrix to a $2^{⟦W'⟧} × 2^{⟦W'⟧}$ density
matrix, where $⟦W⟧$ is the size of a wire type:\footnote{In practice, the
  semantics must be ``padded'' by an additional type so that it can be applied
  in a larger quantum system.}
\[
    ⟦\One⟧ = 0 \qquad ⟦\Qubit⟧ = ⟦\Bit⟧ = 1 \qquad ⟦W_1 ⊗ W_2⟧ = ⟦W_1⟧ + ⟦W_2⟧
\]
In this work we use mixed states only to refer to total, as opposed to partial,
distributions. This means that all mixed states in our semantics have traces equal to $1$.

In this work, we give two different semantics for quantum circuits that differ in how they treat assertions. The \emph{safe} semantics corresponds to an operational model that does not trust assertions, where an \coqe{assertx} gate first measures the input qubit before discarding the result. The \emph{unsafe} semantics assumes that all assertions are accurate, so an \coqe{assertx} gate simply discards its input qubit without measuring it. The two semantics coincide exactly when all assertions in a circuit are accurate, in which case we call the circuit \emph{valid}.

In the safe semantics, assertions are identical to the \coqe{discard} gate,
which measures and then discards the qubit. 
\begin{align*}
\mathtt{denote\_safe} \; \mathtt{U} \; \rho &= \denote{U} \rho \denote{U}^{\dag} \\
\mathtt{denote\_safe} \; \mathtt{init_0} \; \rho &= \ket{0} \rho \bra{0} \\
\mathtt{denote\_safe} \; \mathtt{init_1} \; \rho &= \ket{1} \rho \bra{1} \\
\mathtt{denote\_safe} \; \mathtt{meas} \; \rho &= \ket{0}\bra{0} \rho \ket{0}\bra{0} + \ket{1}\bra{1} \rho \ket{1}\bra{1} \\
\mathtt{denote\_safe} \; \{ \mathtt{discard, \; assert_0, \; assert_1\} } \; \rho  &= \bra{0} \rho \ket{0} + \bra{1} \rho \ket{1}
\end{align*}

Here $\denote{U}$ is the unitary matrix corresponding to the gate \coqe{U};
multiplying by $\denote{U}$ and $\denote{U}^{\dagger}$ is equivalent to applying
$\denote{U}$ to all the pure states in the distribution. The initialization
gates \coqe{init0} and \coqe{init1} both add a single qubit to the system in
the $\ket{0}$ and $\ket{1}$ state respectively. The \coqe{meas} gate produces a
mixed state corresponding to a probability distribution over the measurement
result ($\ket{0}$ or $\ket{1}$). The \coqe{discard} gate removes a
classical-valued bit from the state.

Under the safe semantics, the assertions \coqe{assert0} and \coqe{assert1} are
treated as a measurement followed by a discard. This is semantically the same as the denotation of \coqe{discard}, except that \coqe{discard} is guaranteed by the type system to only throw away a classically valued bit. This operation on qubits is safe (\ie results in a total density matrix) even if the qubit is in a superposition of $\ket{0}$ and $\ket{1}$. %the circuit still corresponds to a sound (though non-unitary) quantum computation.
%due to the implicit measurement. 

% Indeed, we show that our safe denotations correspond to \emph{superoperators},
% functions that take mixed states to mixed states. (Here, we use mixed states
% only to refer to total, rather than partial distributions. Equivalently, all
% of our mixed states have a trace equal to $1$.)

The \emph{unsafe} semantics is the same as the safe semantics, except for \coqe{assert0} and \coqe{assert1}:
\begin{align*}
\mathtt{denote\_unsafe} \; \mathtt{assert_0} \; \rho &= \bra{0} \rho \ket{0} \\
\mathtt{denote\_unsafe} \; \mathtt{assert_1} \; \rho &= \bra{1} \rho \ket{1}
\end{align*}

This is unsafe in the sense that, if $\rho$ isn't in the zero state, then an $\mathtt{assert_0}$ produces a density matrix with a trace less than 1. Operationally, this corresponds to the instruction ``throw away this qubit in the zero state'', which is quantum-mechanically impossible in the general case. However, this semantics corresponds to the intended meaning of \coqe{assert_x} when we know the assertion is true. It also ensures that the composition of \coqe{init_x} with \coqe{assert_x} is equivalent to the identity, which allows us to optimize away qubit initialization and discarding. 
%Finally, the operation is unitary, which is generally desirable for quantum circuits.\robert{re-added last sentence. thoughts?} % Finally, when composed with the
% semantics of a corresponding \coqe{init}, it produces a unitary transformation,
% which is easier to reason about and more useful in practice.

We can define what it means for the ancilla assertions in a circuit to be valid by comparing these two different semantic interpretations.
\begin{coq}
Definition valid_ancillae W (c : Circuit W) : Prop := (denote c = denote_unsafe c).
\end{coq}
An equivalent definition states that the unsafe semantics preserves the trace of its input (which is always $1$) and therefore maps it to a total probability distribution.
\begin{coq}
Definition valid_ancillae' W (c : Circuit W) : Prop :=   
  forall \rho, Mixed_State ρ -> trace (denote_unsafe c ρ) = 1.
\end{coq}
%These definitions are equivalent. 
The second definition follows from the first
because the safe semantics is trace preserving. The first follows from the second since \coqe{denote_unsafe c ρ} corresponds to a sub-distribution of \coqe{denote_safe c ρ}. If its trace is one then they must then represent the same distribution.
%unsafe denotation corresponds to a sub-distribution of the . 
%If this part always has trace $0$, it follows that the denotations are the same. \robert{Jen: Is this clear?}

%The second definition follows
%from the first because the unsafe semantics for \coqe{assert0} and
%\coqe{assert1} preserve the trace exactly when applied to a qubit in the desired
%state. \jennifer{Did I get the order of these right?}

These two definitions precisely characterize what it means for circuits to have
always correct assertions. % Unlike any syntactic condition on circuits, every circuit
% with valid assertions satisfies them and no invalid circuit satisfies them.
In the next section, we define syntactic conditions that are sufficient but not
necessary for validity. Programmers will often write syntactically valid circuits
like those produced by \coqe{compile} function in \cref{sec:compilation}), but when
needed the semantic definition of validity is still available. % Naturally, every
% syntactic condition should itself be shown to guarantee validity in the semantic
% sense.

An important property related to the validity of a circuit is its
\emph{reversibility}. We say that \coqe{c} and \coqe{c'} are \emph{equivalent},
written \coqe{c ≡ c'}, if both their safe and unsafe denotations are equal.
(If \coqe{c} and \coqe{c'} are valid, this is equivalent to
\coqe{denote c = denote c'}, but otherwise it is a stronger claim.)  
%As for validity, reversibility is a semantic condition on circuits, described as follows:
Reversibility says that a circuit has a left and right inverse:
\begin{coq}
Definition reversible {W1 W2} (c : Box W1 W2) : Prop := 
  (exists c', c' ;; c ≡ id_circ) /\ (exists c', c ;; c' ≡ id_circ)
\end{coq}
% This says that there is some other circuit that composes with our circuit to form an identity operation. 
%
In \cref{sec:compilation}, the compiler produces circuits that are their own inverse:
\begin{coq}
Definition self_inverse {W} (c : Box W W) : Prop := c ;; c ≡ id_circ. 
\end{coq}

%Reversible circuit are always valid because both the safe and unsafe semantics are trace-non-increasing.

We can now show that in any reversible circuit, all the ancilla assertions hold.

\begin{lemma} \label{lem:reversible}
If $c$ is reversible, then it is valid.
\end{lemma}
\begin{proof}
  Let $c'$ be $c$'s inverse. By the second definition of validity, it suffices
  to show that the trace of \coqe{denote_unsafe c ρ} is equal to $1$ for every
  initial mixed state $ρ$. We know that the trace of \\ \coqe{denote_unsafe id_circ ρ} is $1$, hence
  \begin{coq}
  1 = trace (denote_unsafe (c;;c') ρ) = trace (denote_unsafe c' (denote_unsafe c ρ))
  \end{coq}
  Because the unsafe semantics is trace-non-increasing, it must be the case that
  the trace of \\ \coqe{denote_unsafe c ρ} is $1$ as well.
\end{proof}

%%% Local Variables:
%%% mode: latex
%%% TeX-master: "main"
%%% End:

\section{Syntactically Valid Ancillae}
\label{sec:guarantees}
Let $c$ be a circuit made up only of classical gates: the initialization gates,
the not gate \code{X}, the controlled-not gate \code{CNOT}, and the Toffoli gate
\code{T}. Let $c'$ be the result of reversing the order of the gates in $c$ and
swapping every initialization with an assertion of the corresponding boolean
value. Then every assertion in $c ;; c'$, where semicolons denote sequencing, is
valid.

Unfortunately, every circuit of this form is also equivalent to the identity
circuit, so as a syntactic condition of validity, this is much too restrictive.
In practice, the quantum oracles discussed in the introduction are mostly
symmetric, but they introduce key pieces of asymmetry to compute meaningful
results. In \reverc, this construction is called the \emph{restricted inverse};
QCL~\cite{Omer2005} and Quipper~\cite{Green2013} take similar approaches.

Let $c$ be a circuit with an equal number of input and output wires whose qubits
can be broken up into two disjoint sets: the first $n$ qubits are called the
\emph{source}, and the last $t$ qubits are called the \emph{target}. That
is, %
\coqe{c : Box (n+t ⨂ Qubit) (n+t ⨂ Qubit)}. The syntactic condition of
\emph{source symmetry} on circuits guarantees that $c$ is the identity on all
source qubits. In addition, it guarantees that assertions are only made on
source qubits with a corresponding initialization.

A classical gate \emph{acts on the qubit i} if it affects the value of that
qubit in an $m$-qubit system: \coqe{X} acts on its only argument, \coqe{CNOT}
acts on its second argument (the target), and Toffoli acts on its third
argument.

The property of \emph{source symmetry} on circuits is defined inductively as
follows:
\begin{itemize}
\item The identity circuit is source symmetric.
\item If \coqe{g} is a classical gate and \coqe{c} is source symmetric, then
  \coqe{g ;; c ;; g} is source symmetric.
\item If \coqe{g} is a classical gate that acts on a qubit in the target and \coqe{c} is
  source symmetric, then both \coqe{g ;; c} and \coqe{c ;; g} are source symmetric.
\item If \coqe{c} is source symmetric and \coqe{i} is in the source of \coqe{c}, then \\
  \coqe{init_at b i ;; c ;; assert_at b i} is source symmetric.
\end{itemize}

The key property of a source symmetric circuit is that it does not affect the
value of its source qubits. We say that a circuit \coqe{c} is a \emph{no-op} at qubit $i$
if, when initialized with a boolean $b$, the qubit is still equal to $b$ after executing the circuit. We could define this as $\denote{c}(\rho_1 \otimes \ket{b}\bra{b} \otimes \rho_2) = \rho_1' \otimes \ket{b}\bra{b} \otimes \rho_2'$ for some $\rho_1, \rho_2, \rho_1', \rho_2'$, but this would require $\rho_1$ and $\rho_2$ (and $\rho_1'$ and $\rho_2$') to be separable, which is an unnecessary restriction. 
Instead, we use the \coqe{valid_ancillae} predicate and say if we initialize an ancilla in state $x$ at $i$, apply $b$, and then assert that $i = x$, our assertion will be valid:
\begin{coq}
Definition noop_on (m k : nat) (c : Box (Qubits (1 + m)) (Qubits (1+m)) : Prop :=
  ∀ b, valid_ancillae (init_at b i ;; c ;; assert_at b i).
\end{coq}
We similarly define a predicate, \coqe{noop_on_source n}, that says that a given
circuit is a no-op on each of its first $n$ inputs.

These inductive definitions allow us to state a number of closely related lemmas
about symmetric circuits:

\begin{lemma} \label{lem:gate-noop}
  If the classical gate $g$ acts on the qubit $k$ and $i ≠ k$, then $g$ is a
  no-op on $i$.
\end{lemma}

\begin{lemma} \label{lem:init-assert}   
  Let \code{c} be a circuit such that \code{c ;; assert\_at b i} is a valid assertion. 
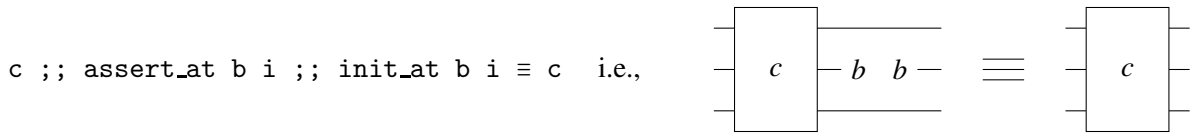
\begin{center}
\begin{minipage}{0.55\textwidth}
\begin{center}
\code{c ;; assert\_at b i ;; init\_at b i ≡ c} \quad \text{\ie}
\end{center}
\end{minipage}\qquad\begin{minipage}{0.38\textwidth}
\begin{tikzpicture}[scale=0.55]
%% assertion and initialization

    \node at (3,1.5) (assert) {b};
    \node at (4,1.5) (init) {b};

%% wires exiting c
    \draw (-0.5,2.5) -- (5,2.5);
    \draw (-0.5,1.5) -- (assert); \draw (init) -- (5,1.5);
    \draw (-0.5,0.5) -- (5,0.5);

%% first circuit c
    \draw[fill=white] (0,0) rectangle (2,3);
    \node at (1,1.5) {c};

%% Equivalence
    \draw (6,1.25) -- (7,1.25);
    \draw (6,1.5) -- (7,1.5);
    \draw (6,1.75) -- (7,1.75);

%% second set of wires
    \draw (8,2.5) -- (11,2.5);
    \draw (8,1.5) -- (11,1.5);
    \draw (8,0.5) -- (11,0.5);

%% second circuit c
    \draw[fill=white] (8.5,0) rectangle (10.5,3);
    \node at (9.5,1.5) {c};
\end{tikzpicture}
\end{minipage}
\end{center}
\end{lemma}

\begin{lemma} \label{lem:inSeq-noop}
  If \code{c} and \code{c'} are both no-ops on qubit $i$, then \code{c ;; c'} is also a no-op on
  qubit $i$.
\end{lemma}

\begin{conjecture} \label{lem:symmetric-noop}
  If \code{c} is source symmetric, then it is a no-op on its source.
\end{conjecture}

These lemmas have been admitted, rather than proven, in the Coq development (Symmetric.v). \cref{lem:symmetric-noop} is labeled as a \emph{conjecture} rather than a lemma, since we do not yet have a paper proof of the statement. It may be the case that we need to strengthen our definition of no-op for this conjecture to hold.

%\begin{proof}
%  By induction on the proof of source symmetry. For the identity, the property
%  is trivial.
% 
%  If \coqe{c} is \coqe{g ;; c' g}, where \coqe{g} acts on qubit $i$, then it suffices to show
%  \coqe{g ;; c ;; g} is a no-op on qubit $i$, since all the others follow from
%  \cref{lem:gate-noop}. For the not gate \coqe{X} the result follows because 
%  \begin{coq}
%init_at b i ;; X at i ;; c ;; X_at i ;; assert_at b i ≡ init_at (¬ b) i ;; c ;; assert_at (¬ b) i
%  \end{coq}
%  \jennifer{I'm actually not at all sure of this property for CNOT and T}
%
%  If \coqe{c} is \coqe{c' ;; g} or \coqe{g ;; c'}, where \coqe{g} acts on qubit $i$ in the target,
%  then to show \coqe{c} is a no-op on the source, it suffices to consider $j ≠ i$,
%  which follows from \cref{lem:gate-noop}.
%
%  Finally, if \coqe{c} is equal to \coqe{init_at b i ;; c' ;; assert_at b i}  then to show
%  \coqe{c} is a no-op on index $j$ in the source, it suffices to notice that
%  \coqe{init_at b i} and \coqe{assert_at b i} are both no-ops on $j$. \jennifer{Actually, this does not really line up. Needs more work.}
%\end{proof}

Since all ancillae in a source symmetric circuit occur on sources, we can prove from the statements above that source symmetric circuits are valid.

\begin{theorem} \label{lem:symmetric-valid}
  If \coqe{c} is source symmetric, then all its assertions are valid.
\end{theorem}

Source symmetric circuits also satisfy a more general property---they are reversible.
The inverse of a source symmetric circuit is defined by induction on source symmetry:
\begin{itemize}
    \item The inverse of the identity circuit is the identity;
    \item The inverse of \coqe{g ;; c ;; g} is \coqe{g ;; c^-1 ;; g};
    \item The inverses of \coqe{c ;; g} and \coqe{g ;; c} are \coqe{g ;; c^-1} and \coqe{c^-1 ;; g}; and
    \item The inverse of \coqe{init_at b i ;; c ;; assert_at b i} is \coqe{init_at b i ;; c^-1 ;; assert_at b i}.
\end{itemize}
Clearly, the inverse of any source symmetric circuit is also source symmetric, and the inverse is involutive, meaning $(c^{-1})^{-1} = c$.

\begin{theorem}
  If \coqe{c} is source symmetric, then \coqe{c^-1 ;; c} is equivalent to the
  identity circuit.
\end{theorem}
\begin{proof}
  By induction on the proof of source symmetry. The only interesting case is the
  case for ancilla, showing
  \begin{coq}
init_at b i ;; c^-1 ;; assert_at b i ;; init_at b i ;; c ;; assert_at b i ≡ id_circ.
  \end{coq}
  From \cref{lem:symmetric-valid} we know that the circuit
  \coqe{init_at b i ;; c^-1 ;; assert_at b i} is valid. Then 
  \cref{lem:init-assert} tells us that
  \coqe{init_at b i ;; c^-1 ;; assert_at b i;; init_at b i} is equivalent to
  \coqe{init_at b i ;; c^-1}.  Thus the goal reduces to \coqe{init_at b i ;; c^-1 ;; c ;; assert_at b i}.
  This is equivalent to the identity by the induction hypothesis as well as the
  fact that \coqe{init_at b i ;; assert_at b i} is the identity.
\end{proof}

%\jennifer{Things left to prove in the Coq formulation}
%\begin{itemize}
%  \item Define toffoli
%  \item Prove that if $g$ acts on $k$ and $i ≠ k$, then $g$ is a no-op on $i$.
%    (easy)
%  \item Prove that $\assert{b}{i} ∘ \init{b}{i} ≡ id$ (easy)
%  \item Prove that if $\assert{b}{i} ∘ c$ is valid, then 
%    $\init{b}{i} ∘ \assert{b}{i} ∘ c ≡ c$. (important, not sure how to prove)
%  \item Change definition of equivalence to refer to both safe and unsafe semantics, and be able to prove properties about validity and composition. (would be good to get a better grasp of what these concepts mean, but this is not the crux of the proofs)
%
%  \item If $g$ is a classical gate and $c$ is a no-op on $i$, then $g ∘ c ∘ g$ is a no-op on its source (very important, don't know how to prove)
%
%  \item If $c$ is a no-op on $i$ in the source of $c$, then
%    $\assert{b}{j} ∘ c ∘ \init{b}{j}$ is a no-op on $i$. (subtle, not sure how
%    to proceed, but possibly not that hard)
%
%\end{itemize}
%\jennifer{End things to prove}

We can now say that any circuit followed by its inverse is valid. But this theorem is easily extensible. For instance, we can add the following to our inductive definition of symmetric, and the theorem will still hold:
\begin{itemize}
\item If \coqe{c} is source symmetric, and \coqe{c ≡ c'}, then \coqe{c'} is source symmetric.
\end{itemize}
This extension allows us to use existing (semantic) equivalences to satisfy our
(syntactic) source symmetry predicate, which in turn proves the semantic
property of validity. For example, because teleportation is semantically
equivalent to the identity circuit, we know trivially that it is valid, even
though it is not source symmetric. The Coq development provides many useful
compiler optimizations in the file \coqe{Equations.v} that can now be used in
establishing source symmetry.

\section{Compiling Oracles}
\label{sec:compilation}
Now that we have syntactic guarantees for circuit validity, we can consider a
compiler from boolean expressions to source-symmetric circuits, producing the
quantum oracles described in the introduction to this chapter. The resulting circuits are all
source symmetric, so it follows from the previous section that they are valid.

We begin with a small boolean expression language, borrowed from Amy \etal \cite{Amy2017}, with variables, constants, negation ($\neg$), conjunction ($\wedge$), and exclusive-or ($\oplus$).
\[ b ::= x \mid \TRUE \mid \FALSE \mid \neg b \mid b_1 \wedge b_2 \mid b_1 \oplus b_2 \]

The \emph{interpretation function} \coqe{[b]_f} takes a boolean
expression \coqe{b} and a valuation function \coqe{f : Var -> bool} and returns the 
value of the boolean expression with the variables assigned as in \coqe{f}.

The compiler, shown in \cref{fig:compiler}, takes a boolean expression \coqe{b}
and a \emph{map} \coqe{Γ} from the variables of $b$ to the wire
indices\footnote{In the Coq development, these maps are represented by linear
  typing contexts.}. The resulting circuit \coqe{compile b Γ} has \coqe{|Γ|+1}
qubit-valued input and output wires, where \coqe{|Γ|} is the number of variables
in the scope of \coqe{b}.

%Fixpoint compile (b : bexp) (Γ : Ctx) : 
%  Box (Qubits (1 + |Γ|)) (Qubits (1 + |Γ|)) :=
%  match b with
%  | b_t          => X_at 0
%  | b_var v      => CNOT_at (1 + get_index Γ v) 0
%  | b_and b1 b2  => init_at false 1                  ;;
%                    id_circ || compile b1 Γ           ;;
%                    init_at false 2                  ;;
%                    id_circ || id_circ || compile b2 Γ ;;
%                    Toffoli_at 1 2 0                 ;;
%                    id_circ || id_circ || compile b2 Γ ;;
%                    assert_at false 2                ;;
%                    id_circ || compile b1 Γ           ;;
%                    assert_at false 1 
%  | ...
%  end.
%\pagebreak
\begin{figure}
\begin{coq}
Fixpoint compile (b : bexp) (Γ : Ctx) : Square_Box (S (⟦Γ⟧) ⨂ Qubit) :=
  match b with
  | b_t          => TRUE || id_circ 
  | b_f          => FALSE || id_circ
  | b_var v      => CNOT_at (1 + index v Γ) 0
  | b_not b      => init_at true 1             ;;
                    id_circ || (compile b Γ)   ;;
                    CNOT_at 1 0                ;;
                    id_circ || (compile b Γ)   ;;
                    assert_at true 1 
  | b_and b1 b2  => init_at false 1            ;;
                    id_circ || compile b1 Γ    ;;
                    init_at false 2            ;;
                    id_circ || id_circ || compile b2 Γ ;;
                    Toffoli_at 1 2 0                   ;;
                    id_circ || id_circ || compile b2 Γ ;;
                    assert_at false 2          ;;
                    id_circ || compile b1 Γ    ;;
                    assert_at false 1 
  | b_xor b1 b2  => init_at false 1            ;;
                    id_circ || compile b1 Γ    ;;
                    CNOT_at 1 0                ;;                    
                    id_circ || compile b1 Γ    ;; 
                    id_circ || compile b2 Γ    ;;
                    CNOT_at 1 0                ;;                    
                    id_circ || compile b2 Γ    ;;
                    assert_at false 1
  end.
\end{coq}
\caption{Compiler from boolean expressions to source symmetric circuits.}
\label{fig:compiler}
\end{figure}

The compiler uses \coqe{init_at}, \coqe{assert_at}, \coqe{X_at}, \coqe{CNOT_at}, 
and \coqe{Toffoli_at} circuits, each of which applies the corresponding gate to
the given index in the list of $n$ wires. 
% We show the compiler below.
%Here we show only the cases for \coqe{true}, variables, and \coqe{b1 ^ b2} since the other cases are analogous.
It makes heavy use of the sequencing (\coqe{;;}) and parallel (\coqe{||})
operators. The TRUE case in \cref{fig:compiler} outputs the exclusive-or of
\coqe{true} with the target wire, which is equivalent to simply negating the
target wire; the \coqe{FALSE} case reduces to the identity. The variable case
\coqe{b_var} applies a CNOT gate from the variable's associated wire to the
target, thereby sharing its value.

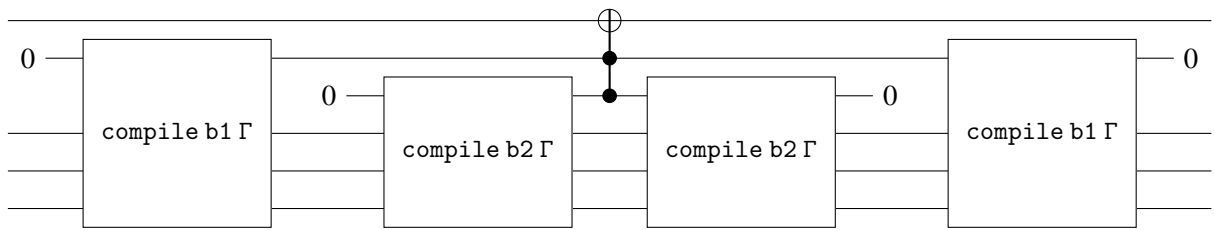
\begin{figure}
\begin{tikzpicture}[scale=1]
%Wires
   \draw(-8,0.25) -- (8,0.25);
   \draw(-8,0.75) -- (8,0.75);
   \draw(-8,1.25) -- (8,1.25);
   \draw(-3.5,1.75) node[left]{$0$} -- (3.5,1.75) node[right]{$0$};
   \draw(-7.5,2.25) node[left]{$0$} -- (7.5,2.25) node[right]{$0$};
   \draw(-8,2.75) -- (8,2.75);

% Compile circuits
    \draw[fill=white] (-7,0) rectangle node {\coqe{compile b1 Γ}} (-4.5,2.5);
   \draw[fill=white] (4.5,0) rectangle node {\coqe{compile b1 Γ}} (7,2.5);
    \draw[fill=white] (-3,0) rectangle node {\coqe{compile b2 Γ}} (-0.5,2);
    \draw[fill=white] (0.5,0) rectangle node {\coqe{compile b2 Γ}} (3,2);

% Toffoli
    
    \draw[thick] (0,1.75) -- (0,2.9);
    \draw (0,2.75) circle [radius=1.5mm];
    \draw[fill=black] (0,2.25) circle [radius=0.9mm];
    \draw[fill=black] (0,1.75) circle [radius=0.9mm];
    
%        \draw[fill=black] (3.5,3.5) circle [radius=1mm];
%    \draw (3.5,3.5) -- (3.5,1.85);
%    \draw (3.5,2) circle [radius=1.5mm];

\end{tikzpicture}
\caption[Compiling $b_1 \wedge b_2$ on 3 qubits]{Compiling $b_1 \wedge b_2$ on 3 qubits. The top wire is the target.}
\label{fig:compile-and}
\end{figure}

The AND case (\cref{fig:compile-and}) is more interesting. We first initialize a qubit in the $0$ state and recursively compile the value of \coqe{b1} to it. We then do the same for \coqe{b2}. We apply a Toffoli gate from \coqe{b1} and \coqe{b2}, now occupying the $1$ and $2$ positions in our list, to the target qubit at $0$. We then reapply the symmetric functions \coqe{compile b2 Γ} and \coqe{compile b1 Γ} to their respective wires, returning the ancillae to their original states and discarding them. We are left with the target wire \coqe{z} holding the boolean value $b_z \oplus (b_1 \wedge b_2)$ and \coqe{|Γ|} wires retaining their initial values. 

Finally, we have the XOR case.  Here we borrow a trick from \reverc~\cite{Amy2017} and allocate only a single ancilla instead of the two we used in the AND case. Instead of calculating $(b_1 \oplus b_2) \oplus t$, where $t$ is the target, we calculate the equivalent $b_2 \oplus (b_1 \oplus t)$, taking advantage of the associativity and commutativity of $\oplus$. Hence, as soon as we've computed $b_1$, we can apply a $CNOT$ from $b_1$ to the target and immediately uncompute $b_1$. This frees up our ancilla, which we then use as a target for \coqe{compile b2}.  

Note that our entire \coqe{compile} circuit is source symmetric, and therefore our assertions are guaranteed to hold by \cref{lem:symmetric-valid}.

We can now go about proving the correctness of this compilation.
\begin{coq}
Theorem compile_correct : forall (b : bexp) (Γ : Ctx) (f : Var -> bool) (z : bool), 
  vars b ⊆ domain Γ ->
  [[compile b Γ]] (bool_to_matrix t ⊗ basis_state Γ f) = 
  bool_to_matrix (z ⊕ [b]_f) ⊗ basis_state Γ f.
\end{coq}

The function \coqe{basis_state} takes the wires referenced by $\Gamma$ and the
assignments of $f$ and produces the corresponding basis state. This forms the
input to the compiled boolean expression along with the target, a classical
qubit in the $\ket{0}$ or $\ket{1}$ state. The statement of \coqe{compile}'s
correctness says that when we apply \coqe{[[compile b Γ]]} to this basis state
with an additional target qubit, we obtain the same matrix with the result of
the boolean expression on the target. The proof follows by induction on the
boolean expression.

\section{Quantum Arithmetic in \qwire}
\label{sec:arithmetic}
In this section, we show how to use the compiler from the previous section to
implement a quantum adder, which has applications in many quantum algorithms,
including Shor's algorithm.
% For instance, Shor's algorithm requires a quantum function that computes power
% of a randomly chosen constant.
A verified quantum adder is therefore an important step towards verifying a
variety of quantum programs.

The input to an adder consists of two $n$-qubit numbers represented as sequences
of qubits $x_{1:n}$ and $y_{1:n}$, as well as a carry-in qubit $c_{in}$. The
output consists of the sum $sum_{1:n}$ and the carry-out $c_{out}$.

To begin, consider a simple 1-bit adder that takes in three bits, $c_{in}$, $x$, and
$y$, and computes their sum and carry-out values. The sum is equal to
$x \oplus y \oplus c_{in}$, and the carry is
$(c_{in} \land (x \oplus y)) \oplus (x \land y)$. The expressions can be
compiled to $4$- and $5$-qubit circuits \coqe{adder_sum} and \coqe{adder_carry},
respectively, where the order of qubits is $c_{out}$, $sum$, $y$, $x$, and
$c_{in}$.
\begin{coq}
Definition adder_sum : Box (4 ⨂ Qubit) (4 ⨂ Qubit) := 
  compile ((c_in ∧ (x ⊕ y)) ⊕ (x ∧ y)) (list_of_Qubits 4).
Definition adder_carry : Box (5 ⨂ Qubit) (5 ⨂ Qubit) := 
  compile (x ⊕ y ⊕ c_in) (list_of_Qubits 5).
Definition adder_1 : Box (5 ⨂ Qubit) (5 ⨂ Qubit) := 
  adder_carry ;; (id_circ || adder_sum).
\end{coq}
Here, \coqe{adder_sum} computes the sum of its three input bits and
\coqe{adder_carry} computes the carry, ignoring the result of \coqe{adder_sum}.
Semantically, the adder should produce the appropriate boolean values; the
operation \coqe{bools_to_matrix} converts a list of booleans to a density
matrix.
\begin{coq}
Lemma adder_1_spec : forall (cin x y sum cout : bool),
  [[adder_1]] (bools_to_matrix [cout; sum; y; x; cin])
= (bools_to_matrix [ cout ⊕ (c_in ∧ (x ⊕ y) ⊕ (x ∧ y)); 
                   ; sum ⊕ (x ⊕ y ⊕ c_in)
                   ; y; x; cin]).
\end{coq}

Next, we extend the 1-qubit adder to $n$ qubits. The $n$-qubit adder contains
two parts---\coqe{adder_left} and \coqe{adder_right}---defined recursively using
padded \coqe{adder_1} and \coqe{adder_carry} circuits. The left part computes
the sum and carry sequentially from the least significant bit, initializing an
ancilla for the carry in each step. When it reaches the most significant bit, it
computes the most significant bit of the sum and carry-out using the 1-qubit
adder. The right part of the adder uncomputes the carries and discards the
ancillae. The definitions of the circuits are shown below and illustrated in
\cref{fig:adder}.
\begin{coq}
Fixpoint adder_left (n : nat) : Box ((1+3*n) ⨂ Qubit) ((1+4*n) ⨂ Qubit) := 
 match n with
 | S n' => (id_circ || (id_circ || (id_circ || (adder_left n')))) ;;
           (init_at false (4*n) 0) ;;
           (adder_1_pad (4*n')) 
 end.
Fixpoint adder_right (n : nat) : Box ((1+4*n) ⨂ Qubit) ((1+3*n) ⨂ Qubit) := 
 match n with
 | O => id_circ
 | S n' => (adder_carry_pad (4*n')) ;;
           (assert_at false (4*n) 0) ;;
           (id_circ || (id_circ || (id_circ || (adder_right n')))) 
 end.
Fixpoint adder_circ (n : nat) : Box ((2+3*n) ⨂ Qubit) ((2+3*n) ⨂ Qubit) := 
 match n with
 | O => id_circ
 | S n' => (id_circ || (id_circ || (id_circ || (id_circ || (adder_left n')))));; 
           (adder_1_pad (4*n')) ;;
           (id_circ || (id_circ || (id_circ || (id_circ || (adder_right n'))))) 
 end.
\end{coq}

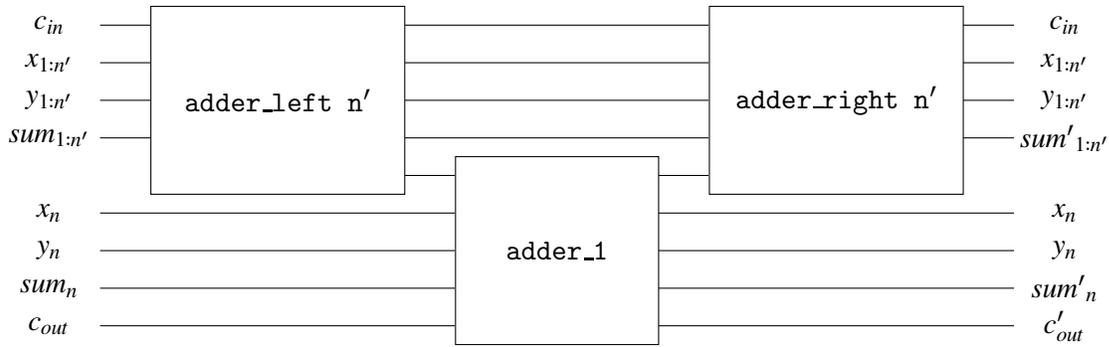
\begin{figure}
\begin{tikzpicture}[y=-0.5cm,x = 1.35cm, scale=1]
  
      % Wires
    \draw (0,0) -- (9,0);
    \draw (0,1) -- (9,1);
    \draw (0,2) -- (9,2);
    \draw (0,3) -- (9,3); 
    \draw (2,4) -- (7,4); % from adder_left
    \draw (0,5) -- (9,5);
    \draw (0,6) -- (9,6);
    \draw (0,7) -- (9,7);
    \draw (0,8) -- (9,8);

        % Nodes

  \draw (-0.5,0) node {$c_{in}$};
  \draw (-0.5,1) node {$x_{1:n'}$};
  \draw (-0.5,2) node {$y_{1:n'}$};
  \draw (-0.5,3) node {$sum_{1:n'}$};
  \draw (-0.5,5) node {$x_n$};
  \draw (-0.5,6) node {$y_n$};  
  \draw (-0.5,7) node {$sum_n$};
  \draw (-0.5,8) node {$c_{out}$};  

  \draw (9.5,0) node {$c_{in}$};
  \draw (9.5,1) node {$x_{1:n'}$};
  \draw (9.5,2) node {$y_{1:n'}$};
  \draw (9.5,3) node {${sum'}_{1:n'}$};

  \draw (9.5,5) node {$x_n$};
  \draw (9.5,6) node {$y_n$};  
  \draw (9.5,7) node {${sum'}_n$};
  \draw (9.5,8) node {$c'_{out}$};

    % Boxes
    \draw[fill=white] (0.5,-0.5) rectangle node {\texttt{adder\_left }$\mathtt{n'}$} (3,4.5);
    \draw[fill=white] (6,-0.5) rectangle node {\texttt{adder\_right }$\mathtt{n'}$} (8.5,4.5);
    \draw[fill=white] (3.5,3.5) rectangle node {\texttt{adder\_1}} (5.5,8.5);

\end{tikzpicture}
\caption[A quantum adder]{A quantum circuit for the $n$-adder where $n' = n-1$ . The $n'$ ancillae created in \coqe{adder_left} are all terminated inside \coqe{adder_right}.}
\label{fig:adder}
\end{figure}

We now prove the correctness of the $n$-qubit adder: 
\begin{coq}
Lemma adder_circ_n_spec : forall (n : nat) (f : Var -> bool),
    let li := list_of_Qubits (2 + 3 * n) in
    [[adder_circ_n n]] (ctx_to_matrix li f)
    = (ctx_to_matrix li (compute_adder_n n f)).
\end{coq}

Like \coqe{bools_to_matrix} above, \coqe{ctx_to_matrix} takes in a context and an assignment $f$ of variables to booleans and constructs the corresponding density matrix. The function \coqe{compute_adder_n} likewise takes a function $f$ that assigns values to each of the $3*n + 2$ input variables and returns a boolean function $f'$ representing the state of the same variables after addition (computed classically). The specification states that the $n$-bit adder circuit computes the state corresponding to the function \coqe{compute_adder_n} for any initial assignment.

Note that the lemma gives a correspondence between the denotation of the circuit and functional computation on the assignment. This can reduce the time required to verify more complex arithmetic circuits. A natural next step is to verify the correspondence between our functions on lists of booleans and Coq's binary representations of natural numbers, thereby grounding our results in the Coq standard library and allowing us to easily move between numerical representations.

\section{Related and Future Work}
\label{sec:related}
The area of reversible computation well predates quantum
computing. % In his ``Logical Reversibility of Computation''~
Bennett~\cite{Bennett1973} introduced the reversible Turing
machine in 1973, with the intent of designing a computer with low energy consumption,
since destroying information necessarily dissipates energy. Toffoli designed
%this with a paper called ``Reversible Computing''~\cite{Toffoli1980} that
the general approach for converting classical circuits to reversible ones
presented in our introduction. While these ideas strongly influenced quantum
computation, reversible computation is a subject of great interest in its own
right, and we refer the interested reader to a standard text on the
subject~\cite{Devos2011,Perumalla2013}.

This work builds heavily on the Quipper quantum programming
language~\cite{Green2013, Green2013a}, which includes ancillae terminations that
are optimized away by joining them to corresponding initializations.
Unfortunately, as is noted in the introduction, the language has no way of
checking its ``assertive terminations'':
\begin{quote}
The first thing to note is that it is the programmer, and not the compiler, who is asserting that the qubit is in state $\ket{0}$ before being terminated. In general, the correctness of such an assertion depends on intricacies of the particular algorithm, and is not something that the compiler can verify automatically. It is therefore the programmer’s responsibility to ensure that only correct assertions are made. The compiler is free to rely on these assertions, for example by applying optimizations that are only correct if the assertions are valid.~\cite{Green2013}
\end{quote}
This work was motivated precisely by the desire to fill in this gap, and by Quipper's demonstration of the power of assertive terminations. 

The other important work in this space is Amy \etal's \reverc~\cite{Amy2017},
which builds upon the \revs programming language~\cite{Parent2017}, a small
heavily-optimized language for reversible computing. \reverc verifies many of
the optimizations from \revs and includes a compiler from boolean expressions
to reversible circuits. The validity of this compilation is verified in the
F$^\star$ programming language~\cite{Swamy2016}. One key challenge in this paper
was to port that compiler from a language that uses only classical operations on
numbered registers (and whose semantics are therefore in terms of boolean
expressions), to a language using higher-order abstract syntax whose denotation
is in terms of density matrices (representing pure and mixed quantum states).

\paragraph{The State of \qwire}

This paper, and the whole \qwire project, is a work in progress. \qwire has been
used to verify some interesting programs, including quantum teleportation, superdense coding,
Deutsch's algorithm and a variety of random number generators (see HOASProofs.v
in the Coq development). It can also be used to prove the validity of a number of circuit
optimizations, such as those of Staton~\cite{Staton2015} (see Equations.v).
However, much remains to be done. The authors' goal is to formally verify all of
the claims in this paper, though some work still remains.

In particularly, the following lemmas remain to be proved in Coq, by section:
\begin{itemize}
\item In \cref{sec:semantics}, the proof of the equivalence of the two
  definitions of \coqe{valid_ancillae} (though the this paper does not build on
  that equivalence); and the proof of \cref{lem:reversible}.
\item In \cref{sec:guarantees},
  \cref{lem:gate-noop,lem:init-assert,lem:inSeq-noop,lem:symmetric-noop}.
\item In \cref{sec:compilation}, that the \coqe{CNOT_at} and \coqe{Toffoli_at}
  circuits, as well as sequencing \coqe{;;} and parallel \coqe{∥} combinators,
  match their intended semantics.
%\item In \cref{sec:arithmetic}, the proof of the adder's correctness on $n$ qubits.
\end{itemize}

The next step for \qwire is to implement and verify circuit optimizations. We already have a number
of equivalences we can in principle use to rewrite our circuits, and this work introduces new possible 
optimization, like reusing ancillae. It also allows us to treat circuits that properly initialize and dispose of ancilla as unitary circuits, allowing for further optimizations. Given that much of the progress towards practical quantum computing comes from increasingly clever optimizations (in tandem with more powerful quantum computers), verified compilation should play an important and exciting role in the near future. 

We can also expand our boolean expression language to allow us to compile a broader range of classical functions. This would allow us to program the adder from the previous section entirely within this language and then compile it to a quantum circuit. Such a program would require the following features:
\begin{enumerate}
\item Pairs would allow us to represent binary numbers, where \coqe{(true,(true,(false,true)))} could represent $1101$ or $13$. 
\item Projection operators would allow us to extract values from pairs.
\item Let bindings would allow us to reuse sub-circuits for efficiency.
\end{enumerate}
We could then make our \coqe{bexps} dependently typed, which would allow us to associate \coqe{bexps} with the number of wires entering and exiting the corresponding circuit. We could even include types for $n$-bit numbers that correspond to product types. And naturally, there is much more that we could do with the \coqe{bexp} language, including adding lambdas, branching, recursion, and other common programming language idioms.

We would also like to optimize our current compiler. Our \coqe{compile} function borrows a trick from \reverc~\cite{Amy2017}, in that it doesn't use additional ancilla to compile exclusive-ors. However, there are many optimizations that remain to be done, and, given the limitations of today's quantum computers, they are worth implementing.

Finally, a recent innovation in the area of quantum computing concerns so-called \emph{dirty ancillae}. We call an ancilla ``dirty" if it may be initialized in an arbitrary state, not only $\ket{0}$. Haner \etal \cite{Haner2016} show that these can take the place of our ``clean'' ancillae in many quantum circuits, and Q\# \cite{Svore2018} allows us to ``borrow'' a qubit, use it, and then return it to its initial state. Extending the work in this paper to verify that dirty ancillae are uncomputed would require substantial additional machinery but it would have a significant payoff in terms of expressivity without sacrificing reliability.

%%% Local Variables:
%%% mode: latex
%%% TeX-master: "main"
%%% End:

%\section{Related and Future Work}
%\label{sec:conclusion}
%\input{conclusion}

% \bibliographystyle{abbrvnat}
% \bibliography{bibliography}

%\nocite{*}
\bibliographystyle{eptcsstyle/eptcs}
{\small
\bibliography{biblio}
}

\end{document}